\documentclass[letterpaper, 10 pt, conference]{ieeeconf}  % Comment this line out if you need a4paper
\IEEEoverridecommandlockouts                              % This command is only needed if 
\usepackage{00_ps}

\preprinttrue% Either \preprinttrue or \preprintfalse
\publishedfalse% Either \publishedtrue or \publishedtrue

\renewcommand{\copyrightOwnerFull}{IEEE}
\renewcommand{\copyrightOwnerShort}{IEEE}

\glsdisablehyper % disable hyperlinks for glossary entries
\DeclareAcronym{ap}{
    short = AP,
    long  = Action Propagation,
}
\DeclareAcronym{bf}{
    short = BF,
    long  = Barrier Function,
}
\DeclareAcronym{c2c}{
    short = C2C,
    long  = Center-to-Center,
}
\DeclareAcronym{cav}{
    short = CAV,
    long  = Connected and Automated Vehicle,
}
\DeclareAcronym{cbf}{
    short = CBF,
    long  = Control Barrier Function,
}
\DeclareAcronym{cg}{
    short = CG,
    long = Center of Gravity,
    short-plural = s,
    long-plural-form = Centers of Gravity,
}
\DeclareAcronym{clf}{
    short = CLF,
    long  = Control Lyapunov Function,
}
\DeclareAcronym{cpm}{
    short = CPM,
    long  = Cyber-Physical Mobility
}
\DeclareAcronym{cpmlab}{
    short = CPM Lab,
    long  = Cyber-Physical Mobility Lab
}
\DeclareAcronym{dcbf}{
    short = DCBF,
    long  = Discrete-Time Control Barrier Function,
}
\DeclareAcronym{ecbf}{
    short = ECBF,
    long  = Exponential \ac{cbf},
    short-indefinite = an,
}
\DeclareAcronym{hocbf}{
    short = HOCBF,
    long  = High-Order \ac{cbf},
    short-indefinite = an,
}
\DeclareAcronym{il}{
    short = IL,
    long  = Imitation Learning,
    short-indefinite = an,
}
\DeclareAcronym{mappo}{
    short = MAPPO,
    long  = Multi-Agent \ac{ppo},
    short-indefinite = an,
}
\DeclareAcronym{maddpg}{
    short = MADDPG,
    long  = Multi-Agent Deep Deterministic Policy Gradient,
    short-indefinite = an,
}
\DeclareAcronym{mas}{
    short = MAS,
    long  = Multi-Agent System,
    short-indefinite = an,
}
\DeclareAcronym{mdp}{
    short = MDP,
    long  = Markov decision process,
    short-indefinite = an,
}
\DeclareAcronym{mg}{
    short = MG,
    long  = Markov Game,
    short-indefinite = an,
}
\DeclareAcronym{ml}{
    short = ML,
    long  = Machine Learning,
    short-indefinite = an,
}
\DeclareAcronym{mtv}{
    short = MTV,
    long  = Minimum Translation Vector,
    short-indefinite = an,
}
\DeclareAcronym{mpc}{
    short = MPC,
    long  = model predictive control,
    short-indefinite = an,
}
\DeclareAcronym{marl}{
    short = MARL,
    long  = Multi-Agent Reinforcement Learning,
    short-indefinite = an,
}
\DeclareAcronym{ocp}{
    short = OCP,
    long  = Optimal Control Problem,
    short-indefinite = an,
    long-indefinite = an,
}
\DeclareAcronym{per}{
    short = PER,
    long  = Prioritized Experience Replay
}
\DeclareAcronym{pomdp}{
    short = POMDP,
    long  = Partially Observable \ac{mdp}
}
\DeclareAcronym{pomg}{
    short = POMG,
    long  = Partially Observable \ac{mg}
}
\DeclareAcronym{ppo}{
    short = PPO,
    long  = Proximal Policy Optimization
}
\DeclareAcronym{qp}{
    short = QP,
    long  = Quadratic Program,
}
\DeclareAcronym{rhc}{
    short = RHC,
    long  = receding horizon control,
    short-indefinite = an,
}
\DeclareAcronym{rl}{
    short = RL,
    long  = Reinforcement Learning,
    short-indefinite = a,
}
\DeclareAcronym{sat}{
    short = SAT,
    long = Separating Axis Theorem,
    short-indefinite = an
}
\DeclareAcronym{ttcbf}{
    short = TTCBF,
    long  = Truncated Taylor \ac{cbf},
}
\DeclareAcronym{zsg}{
    short = ZSG,
    long  = Zero-Shot Generalization,
}
\input{00_copyright.tex}

\begin{document}
\title{
    High-Order Control Barrier Functions: Insights and a Truncated Taylor-Based Formulation
    \thanks{This research was supported by the Bundesministerium für Digitales und Verkehr (German Federal Ministry for Digital and Transport) within the project ``Harmonizing Mobility'' (grant number 19FS2035A).}
}

\author{
    Jianye Xu$^{1}$\,\orcidlink{0009-0001-0150-2147},~\IEEEmembership{Student~Member,~IEEE},
    Bassam Alrifaee$^{2}$\,\orcidlink{0000-0002-5982-021X},~\IEEEmembership{Senior Member, ~IEEE}% <-this % stops a space
    \thanks{$^{1}$The author is with the Department of Computer Science, RWTH Aachen University, Germany, \texttt{xu@embedded.rwth-aachen.de}}
    \thanks{$^{2}$The author is with the Department of Aerospace Engineering, University of the Bundeswehr Munich, Germany, \texttt{bassam.alrifaee@unibw.de}}
}
    \maketitle
\thispagestyle{IEEEtitlepagestyle}
% \pagestyle{plain}
% ! Abstract
\begin{abstract}
We examine the complexity of the standard High-Order Control Barrier Function (HOCBF) approach and propose a truncated Taylor-based approach that reduces design parameters. First, we derive the explicit inequality condition for the HOCBF approach and show that the corresponding equality condition sets a lower bound on the barrier function value that regulates its decay rate. Next, we present our Truncated Taylor CBF (TTCBF), which uses a truncated Taylor series to approximate the discrete-time CBF condition. While the standard HOCBF approach requires multiple class $\mathcal{K}$ functions, leading to more design parameters as the constraint's relative degree increases, our TTCBF approach requires only one. We support our theoretical findings in numerical collision-avoidance experiments and show that our approach ensures safety while reducing design complexity.
\par\medskip
\end{abstract}
\acresetall

% Keywords: Control barrier function, Lyapunov methods, safety-critical control

%===============================================================================
\section{Introduction}\label{sec:introduction}
\Acp{cbf} have received increasing attention in recent years as a way to ensure safety in dynamical systems \cite{ames2014control}. The main idea behind \acp{cbf} is to guarantee that the system state stays within a predefined safe set by enforcing suitable constraints on the control inputs. By maintaining the forward invariance of this safe set, the system remains safe.

A key component in \acp{cbf} is class $\mathcal{K}$ functions, which regulate how safety constraints are enforced. Specifically, they control the decay rate of the \ac{cbf} value. The tuning of these functions strongly affects the balance between conservatism and aggression. A conservative choice can overly restrict system performance and impair the feasibility of the optimization problem, whereas an aggressive choice can allow excessive risk and potentially compromise safety \cite{zeng2021safetycriticala}. The example in \cref{fig_example_class_k} illustrates how different class $\mathcal{K}$ functions influence a robot's obstacle-avoidance behavior, ranging from slow and conservative to rapid and aggressive maneuvers. Note that conservatism in \acp{cbf} can also caused by factors like uncertainty handling \cite{alan2023parameterized} and over-approximation of the safe set \cite{xu2024learningbased}, which are out of our scope.

\begin{figure}[t]
    \centering
    \includegraphics[width=1.0\linewidth]{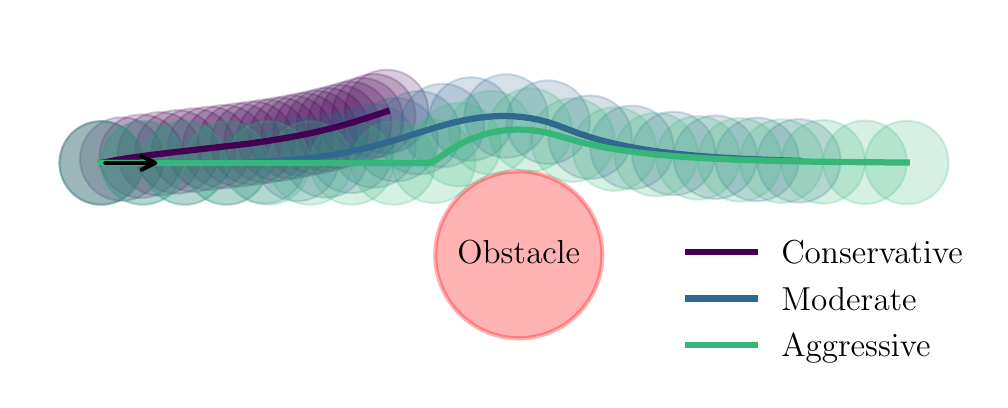}
    \caption{An obstacle-avoidance example with three different class $\mathcal{K}$ functions: conservative, moderate, and aggressive. Footprints in circles; trajectories in solid lines.}
    \label{fig_example_class_k}
\end{figure}

The standard \ac{cbf} approaches cannot handle constraints with high relative degrees. The relative degree of a constraint is the number of times one must differentiate it along the system dynamics before the control input appears. Early solutions include a backstepping-based approach \cite{hsu2015control} and exponential \acp{cbf} with a virtual input-output linearization method \cite{nguyen2016exponential}. Subsequently, the \ac{hocbf} approach introduced in \cite{xiao2019control} provided a systematic way to extend standard \ac{cbf} formulations by chaining high-order derivatives. However, it uses multiple class $\mathcal{K}$ functions, equaling the relative degree of the candidate \ac{cbf}. This complicates control design, particularly for constraints with high relative degrees, since each class $\mathcal{K}$ function requires tuning. Various studies have been proposed to reduce the effort in the design of the class $\mathcal{K}$ functions in \acp{hocbf}, including penalty methods \cite{xiao2019control}, adaptive \ac{cbf} \cite{xiao2022adaptive, xiong2023discretetime}, and learning-based class $\mathcal{K}$ functions \cite{ma2022learning, kim2025learning}.

In this work, we offer new insights into the \ac{hocbf} approach from \cite{xiao2019control} by deriving the explicit form of its safety condition and showing the imposed lower bound on the \ac{cbf} value. This bound provides insights into the resulting controller and can serve to guide the tuning of class $\mathcal{K}$ functions. Furthermore, we propose an alternative approach to handle constraints with high relative degrees, termed \ac{ttcbf}. It simplifies control design by using only one class $\mathcal{K}$ function.

\Cref{sec:preliminaries} reviews \acp{cbf} and \acp{hocbf}.
\Cref{sec:main} presents the main results, including the analytical lower bound imposed by \acp{hocbf} and our approach.
\Cref{sec:experiments} validates the main results numerically.
\Cref{sec:discussions} discusses our approach, and
\cref{sec:conclusions} concludes.

%===============================================================================
%===============================================================================
\section{Preliminaries} \label{sec:preliminaries}
In this section, we briefly revisit \acp{cbf} and \acp{hocbf} for continuous‐time and discrete‐time systems.

%==========================
% \subsection{Control Barrier Functions} \label{sec:cbf}
We consider an input-affine control system given by
\begin{equation} \label{eq:affine-sys}
    \dot{\bm{x}} = f(\bm{x}) + g(\bm{x})\,\bm{u},
\end{equation} 
where the system state $\bm{x} \in \mathcal{X} \subset \mathbb{R}^n$ and the control input $\bm{u} \in \mathcal{U} \subseteq \mathbb{R}^m$. Functions $f: \mathbb{R}^n \to \mathbb{R}^n$ and $g: \mathbb{R}^n \to \mathbb{R}^{n \times m}$ are globally Lipschitz.

\begin{definition}[Class $\mathcal{K}$ functions]
    A continuous function $\alpha: [0,a) \to [0,\infty)$ (with $a>0$) is said to belong to class~$\mathcal{K}$ if it is strictly increasing and $\alpha(0)=0$.   
\end{definition}

\begin{definition}[Forward invariant set] \label{def:forward-invariant}
    A set \(C \subset \mathcal{X}\) is forward invariant for system \labelcref{eq:affine-sys} if its solution starting at any \(\bm{x}(t_0) \in C\) satisfies $\bm{x}(t) \in C, \forall t \ge t_0$.
\end{definition}

For a continuously differentiable function $h: \mathcal{X} \to \mathbb{R}$, let 
\begin{equation} \label{eq:safe-set}
    \begin{aligned}
        C &= \{\bm{x} \in \mathcal{X}: h(\bm{x}) \ge 0\}.
    \end{aligned}
\end{equation}

\begin{definition}[\acp{cbf} \cite{ames2019control}]
    Given a set $C$ as in \eqref{eq:safe-set}, a continuously differentiable function $h(\bm{x}): \mathcal{X} \to \mathbb{R}$ is a candidate \ac{cbf} for system \eqref{eq:affine-sys} on $C$ if there exists a class $\mathcal{K}$ function $\alpha$ such that 
    \begin{equation}\label{eq:cbf-cond-ct}
        \sup_{\bm{u} \in \mathcal{U}}\,\big[ \dot{h}(\bm{x},\bm{u}) + \alpha \bigl( h(\bm{x}) \bigr)\big] \ge 0, \quad \forall \bm{x}\in C,
    \end{equation}
    where $\dot{h}(\bm{x},\bm{u})$ denotes the time derivative of $h$ and can also be expressed as $L_f h(\bm{x}) + L_g h(\bm{x})\,\bm{u}$, with $L_f h$ and $L_g h$ denoting the Lie derivatives of $h$ along $f$ and $g$, respectively.
\end{definition}

Now, consider system \labelcref{eq:affine-sys} in the discrete-time domain given by
\begin{equation} \label{eq:affine-sys-discrete}
    \bm{x}_{k+1} = f(\bm{x}_{k}) + g(\bm{x}_{k})\,\bm{u}_{k},
\end{equation}
where the subscript $k \in \mathbb{N}$ denotes the time step. 
% For simplicity, we denote system \labelcref{eq:affine-sys-discrete} as $\dot{\bm{x}}_{k+1} = f(\bm{x}_k,\bm{u}_k)$.

\begin{definition}[Discrete-Time \acp{cbf} \cite{agrawal2017discrete}] \label{def:d-cbf}
    Given a set $C$ as in \eqref{eq:safe-set}, a continuous function $h(\bm{x}): \mathcal{X} \to \mathbb{R}$ is a candidate discrete-time \ac{cbf} for system \eqref{eq:affine-sys-discrete} on $C$ if there exists a class $\mathcal{K}$ function $\alpha$ satisfying $\alpha(z) \le z$ such that
    \begin{equation} \label{eq:cbf-cond-dt}
        \sup_{\bm{u} \in \mathcal{U}}\,\big[\, \Delta h \bigl(\bm{x}_k,\bm{u}_k\bigr) + \alpha \bigl( h(\bm{x}_{k}) \bigr) \,\big] \;\ge 0, \quad \forall \bm{x}_{k}\in C.
    \end{equation}
\end{definition}

One common variant is the \textit{discrete-time exponential \acp{cbf}} \cite{agrawal2017discrete}, in which a linear class $\mathcal{K}$ function is used, i.e., $\alpha(h(\bm{x}_k)) = \lambda \, h(\bm{x}_k)$, where $\lambda \in (0, 1]$ denotes its parameter. Substituting $\Delta h \bigl(\bm{x}_k,\bm{u}_k\bigr)$ in \labelcref{eq:cbf-cond-dt} with forward difference $h(\bm{x}_{k+1}) - h(\bm{x}_{k})$ and rearranging yields 
\begin{equation} \label{eq:cbf-cond-dte}
    \sup_{\bm{u} \in \mathcal{U}}\,\big[\,h\bigl( \bm{x}_{k+1} \bigr) - (1-\lambda)\,h(\bm{x}_{k})\,\big] \;\ge 0, \quad \forall \bm{x}_{k}\in C.
\end{equation}
At each time step $k \in \mathbb{N}$, the control input enforces $h(\bm{x}_{k+1}) \ge (1-\lambda)\,h(\bm{x}_k)$, rendering $h(\bm{x}_k)$ may decrease by at most the factor $1-\lambda$. Applying this recursively yields the imposed lower bound
\begin{equation} \label{eq:lower-bound-de-cbf}
    h(\bm{x}_{k}) \ge (1-\lambda)^{k-k_0} h(\bm{x}_{k_0}), \quad \forall k \ge k_0,
\end{equation}
where $\bm{x}_{k_0} \in \mathcal{X}$ denotes the initial state at time step $k_0$. This lower bound is an exponential function in time step $k$ and hence the name exponential \acp{cbf}. It becomes 
\begin{equation} \label{eq:lower-bound-ce-cbf}
    h(\bm{x}) \ge e^{-\lambda(t-t_0)} h_{\bm{x}_{t_0}}, \quad \forall t \ge t_0
\end{equation}
in the continuous-time domain, where $\bm{x}_{t_0} \in \mathcal{X}$ denotes the initial state at time $t_0$, and $\lambda > 0$ (instead of $\lambda \in (0,1]$).

%==========================
% \subsection{High-Order Control Barrier Functions} \label{sec:hocbf}
In practice, \acp{cbf} often have high \emph{relative degrees}. The relative degree of a function is the number of times one must differentiate it along the system dynamics until the control input appears explicitly. 
\begin{definition}[Relative degree]
    A continuously differentiable function \( h: \mathcal{X} \to \mathbb{R} \) is said to have relative degree \( r \in \mathbb{N} \) with respect to system \labelcref{eq:affine-sys} or \labelcref{eq:affine-sys-discrete} if $\forall \bm{x} \in \mathcal{X}$, \( L_g L_f^{i} h(\bm{x}) = 0, \forall i \in \{0, 1, \ldots, r-2\} \) and \( L_g L_f^{r-1} h(\bm{x}) \neq 0 \).
\end{definition}
Since we use \acp{cbf} to define constraints, we use the relative degrees of \acp{cbf} and relative degrees of constraints interchangeably. If $r>1$, the standard \ac{cbf} condition \eqref{eq:cbf-cond-ct} or \eqref{eq:cbf-cond-dt} fails to capture the control input $\bm{u}$ (because $L_g h(\bm{x})=0$). Study \cite{xiao2019control} proposed \acp{hocbf} to address this limitation. Recursively, define auxiliary functions $\Psi_0(\bm{x}) \coloneqq h(\bm{x})$ and  
\begin{equation} \label{eq:chain-hocbf}
    \Psi_{i}(\bm{x}) \coloneqq  \dot{\Psi}_{i-1}(\bm{x}) + \alpha_i(\Psi_{i-1}(\bm{x})), \quad \forall i = \{ 1, \dots, r \},
\end{equation}
where each $\alpha_i$ is a differentiable class $\mathcal{K}$ function. 
Let
\begin{equation} \label{eq:high-order-sets}
    C_{i} \coloneqq \bigl\{\bm{x} \in \mathcal{X} : \Psi_{i-1} ( \bm{x} ) \ge 0 \bigr\}, \quad \forall i \in\{1, \ldots, r \}.
\end{equation}
Note that by construction, $C_1$ is the original set where we want to ensure forward invariance.

\begin{definition}[\acp{hocbf} \cite{xiao2019control}] \label{def:continuous-t-hocbf}
    Given sets $C_i$ as in \eqref{eq:high-order-sets}, a continuously differentiable function \( h: \mathcal{X} \to \mathbb{R} \) is a candidate \ac{hocbf} with relative degree $r$ for system \labelcref{eq:affine-sys} if there exist continuously differentiable class $\mathcal{K}$ functions $\alpha_i, \forall i\in \{1,...,r\}$, such that 
    \begin{equation} \label{eq:hocbf-cond}
        \sup_{\bm{u} \in \mathcal{U}}\,\big[ \; \Psi_{r}(\bm{x}, \bm{u}) \; \big] \ge 0, \quad \forall \bm{x} \in \bigcap_{i=1}^{r} C_i.
    \end{equation}
\end{definition}

As shown in \cite{xiao2019control}, enforcing \labelcref{eq:hocbf-cond} inductively ensures $\Psi_{i} \ge 0, \forall i \in \{ 0, \ldots, r-1 \}$, thereby rendering set $C_1$ (and $\bigcap_{i=2}^{r} C_i$) forward invariant. A common variant of \acp{hocbf} is \textit{exponential \acp{hocbf}}, which adopt linear class $\mathcal{K}$ functions when defining auxiliary functions \labelcref{eq:chain-hocbf}, i.e., $\alpha_i\bigl( \Psi_{i-1}(\bm{x}) \bigr) = \lambda_i \Psi_{i-1}(\bm{x})$, where $\lambda_i > 0 $ denotes the parameter of the linear class $\mathcal{K}$ function $\alpha_i$.

%===============================================================================
\section{Main Results} \label{sec:main}
In \cref{sec:insights-hocbf}, we provide insights into the \ac{hocbf} approach introduced by \cite{xiao2019control}. In \cref{sec:taylor-based-alternative}, we propose an alternative approach based on the truncated Taylor series to handle \acp{cbf} with high relative degrees.

\subsection{Insights Into \Acp{hocbf}} \label{sec:insights-hocbf}
Our insights include the explicit form of the condition imposed by exponential \acp{hocbf} and the analytical lower bound imposed by this condition.

\begin{lemma} \label{lemma:ct-hocbf-cond-flatten}
    Consider an \ac{hocbf} $h$ in \cref{def:continuous-t-hocbf} with relative degree \(r\). For each \(i \in \{1, \ldots, r\}\), the corresponding auxiliary function defined in \eqref{eq:chain-hocbf} with a linear class $\mathcal{K}$ function can be flattened and equivalently expressed as
    \begin{equation} \label{eq:ct-hocbf-cond-flatten}
    \Psi_{i}\bigl(\bm{x}\bigr) \equiv \sum_{j=0}^{i} e_{i-j}\bigl( \lambda_{1}, \ldots, \lambda_{i} \bigr) \, h^{(j)}\bigl(\bm{x}\bigr).
    \end{equation}
    Here, $e_{i-j} \bigl( \lambda_{1}, \ldots, \lambda_{i} \bigr) \coloneqq \sum_{\substack{I \subseteq \{1, \ldots, i \} \\ | I |=i-j}} \prod_{k \in I} \lambda_{k}$ denotes an elementary symmetric polynomial\footnote{E.g., if \(i=3\), then $e_{0}(\lambda_{1},\lambda_{2},\lambda_{3}) = 1, e_{1}(\lambda_{1},\lambda_{2},\lambda_{3}) = \lambda_{1} + \lambda_{2} + \lambda_{3}, e_{2}(\lambda_{1},\lambda_{2},\lambda_{3}) = \lambda_{1}\lambda_{2} + \lambda_{1}\lambda_{3} + \lambda_{2}\lambda_{3}, e_{3}(\lambda_{1},\lambda_{2},\lambda_{3}) = \lambda_{1}\lambda_{2}\lambda_{3}$.
    Note that \(e_{0}(\cdot)=1\) by convention.}, \(\lambda_{1}, \ldots, \lambda_{r}\) denote the parameters of the linear class \(\mathcal{K}\) functions, and $h^{(i)}$ denotes the $i$th time derivative of the candidate \ac{cbf} $h$.
\end{lemma}

\begin{proof}
    We proceed by induction on \(i\).
    
    \textbf{Base Case (\(i=1\))}: By definition, $\Psi_{1}\bigl(\bm{x}\bigr) \coloneqq \dot{h}\bigl(\bm{x}\bigr) + \lambda_{1}\, h\bigl(\bm{x}\bigr)$.
    Since \(e_{1}(\lambda_{1}) = \lambda_{1}\) and \(e_{0}(\lambda_{1}) = 1\), we have $\Psi_{1}\bigl(\bm{x}\bigr) \equiv e_{1}(\lambda_{1})\, h\bigl(\bm{x}\bigr) + e_{0}(\lambda_{1})\, \dot{h}\bigl(\bm{x}\bigr)$
    
    \textbf{Inductive Step}: Assume for some \(i \in \{1, \ldots, r-1\}\),
    \begin{equation} \label{eq:induction-hypothesis-ct}
    \Psi_{i}\bigl(\bm{x}\bigr) \equiv \sum_{j=0}^{i} e_{i-j}\bigl( \lambda_{1}, \ldots, \lambda_{i} \bigr) \, h^{(j)}\bigl(\bm{x}\bigr),
    \end{equation}
    and we need to show 
    \begin{equation} \label{eq:induction-need-to-show}
    \Psi_{i+1}\bigl(\bm{x}\bigr) \equiv \sum_{j=0}^{i+1} e_{i+1-j}\bigl( \lambda_{1}, \ldots, \lambda_{i+1} \bigr) \, h^{(j)}\bigl(\bm{x}\bigr).
    \end{equation}
    The time derivative of \eqref{eq:induction-hypothesis-ct} is given by $\dot{\Psi}_{i}\bigl(\bm{x}\bigr) \equiv \sum_{j=0}^{i} e_{i-j}\bigl( \lambda_{1}, \ldots, \lambda_{i} \bigr) \, h^{(j+1)}\bigl(\bm{x}\bigr)$. By definition in \eqref{eq:chain-hocbf}, $\Psi_{i+1}\bigl(\bm{x}\bigr) \coloneqq \dot{\Psi}_{i}\bigl(\bm{x}\bigr) + \lambda_{i+1}\,\Psi_{i}\bigl(\bm{x}\bigr)$. Substituting $\dot{\Psi}_i$ and $\Psi_i$, we obtain $\Psi_{i+1}\bigl(\bm{x}\bigr) \equiv \sum_{j=0}^{i} e_{i-j}\bigl( \lambda_{1}, \ldots, \lambda_{i} \bigr) \, h^{(j+1)}\bigl(\bm{x}\bigr) + \lambda_{i+1}\sum_{j=0}^{i} e_{i-j}\bigl( \lambda_{1}, \ldots, \lambda_{i} \bigr) \, h^{(j)}\bigl(\bm{x}\bigr).$ Consider the property of the elementary symmetric polynomial that $e_{k}\bigl( \lambda_{1}, \ldots, \lambda_{i+1} \bigr) = e_{k}\bigl( \lambda_{1}, \ldots, \lambda_{i} \bigr) + \lambda_{i+1}\, e_{k-1}\bigl( \lambda_{1}, \ldots, \lambda_{i} \bigr), \forall k \in \{0, \ldots, i+1\}$, and note the convention \(e_{i+1}(\lambda_{1}, \ldots, \lambda_{i}) = 0\) and \(e_{-1}(\lambda_{1}, \ldots, \lambda_{i}) = 0\), one easily obtains \eqref{eq:induction-need-to-show}.
\end{proof}

\begin{theorem} \label{thm:forward-invariance-hocbf}
    Let $h$ be an \ac{hocbf} as in \cref{def:continuous-t-hocbf} with relative degree \(r\) and associated sets $C_i$ as in \labelcref{eq:high-order-sets}, and let class $\mathcal{K}$ functions in \eqref{eq:chain-hocbf} be linear with parameters $\lambda_i, \forall i \in \{ 1, \ldots, r \}$. Any Lipschitz continuous controller satisfying
    \begin{equation} \label{eq:equiv-hocbf-condition}
    \sup_{\bm{u} \in \mathcal{U}} \Bigl[ \sum_{j=0}^{r} e_{r-j}\bigl( \lambda_{1}, \ldots, \lambda_{r} \bigr) \, h^{(j)}\bigl(\bm{x}\bigr) \Bigr] \ge 0
    \end{equation}
    renders $C_1$ (and $\bigcap_{i=2}^{r} C_i$) forward invariant for system \eqref{eq:affine-sys}.
\end{theorem}

\begin{proof}
    Following \cref{lemma:ct-hocbf-cond-flatten}, the left side of \labelcref{eq:equiv-hocbf-condition} is equivalent to the last auxiliary function $\Psi_r(\bm{x}, \bm{u})$ in \acp{hocbf}. Satisfying \labelcref{eq:equiv-hocbf-condition} is equivalent to satisfying the standard \ac{hocbf} condition \eqref{eq:hocbf-cond}. The remaining proof follows \cite[Thm. 5]{xiao2019control}.
\end{proof}

\begin{corollary} \label{cor:ct-lower-bound-hocbf}
    Consider an exponential \ac{hocbf} $h(t)$ with relative degree \( r \) and initial conditions
    % Consider a continuously differentiable function \( h\bigl(\bm{x}(t)\bigr) \) with relative degree \( r \) that satisfies the exponential \ac{hocbf} condition in the flattened form \labelcref{eq:equiv-hocbf-condition} and the 
    \begin{equation} \label{eq:initial-condition-hocbf}
    \bm{h}_{\mathrm{init}} \coloneqq \bigl[ h_{t_0}, \dot{h}_{t_0}, \ldots, h^{(r-1)}_{t_0} \bigr]^\top \in \mathbb{R}^r
    \end{equation}
    at an initial time $t_0$. The lower bound of $h(t)$, denoted as $h_{\mathrm{lb}}(t)$, imposed by the exponential \ac{hocbf} is given by 
    \begin{equation} \label{eq:lower-bound-hocbf-general-form}
    h_{\mathrm{lb}}(t) \coloneqq \sum_{i=1}^{r} c_i\, e^{-\lambda_i\, (t - t_0)}, \quad \forall t \ge t_0,
    \end{equation}
    subject to the same initial conditions $\bm{h}_{\mathrm{init}}$ in \labelcref{eq:initial-condition-hocbf}. The coefficient \( \bm{c} \coloneqq \bigr[ c_1, \ldots, c_r \bigr]^\top \in \mathbb{R}^r \) is deterministically determined by solving
    \begin{equation} \label{eq:coefficients}
        M_{\lambda} \; \bm{c} = \bm{h}_{\mathrm{init}},
    \end{equation}
    where
    \begin{equation}
        M_{\lambda} =
        \begin{bmatrix}
        (-\lambda_1)^0 & (-\lambda_2)^0 & \cdots & (-\lambda_r)^0 \\
        (-\lambda_1)^1 & (-\lambda_2)^1 & \cdots & (-\lambda_r)^1 \\
        \vdots & \vdots & \ddots & \vdots \\
        (-\lambda_1)^{r-1} & (-\lambda_2)^{r-1} & \cdots & (-\lambda_r)^{r-1}
        \end{bmatrix} \in \mathbb{R}^{r \times r}
    \end{equation}
    is a Vandermonde-like matrix, and $\lambda_1, \ldots, \lambda_r$ are the parameters of the linear class $\mathcal{K}$ functions.
\end{corollary}

\begin{proof}
    Following \cref{thm:forward-invariance-hocbf}, the condition imposed by the exponential \ac{hocbf} is equivalent to \labelcref{eq:equiv-hocbf-condition}. Observe that the equality form of \labelcref{eq:equiv-hocbf-condition}, i.e.,
    \begin{equation} \label{eq:equiv-hocbf-condition-equality}
    \sum_{j=0}^{r} e_{r-j}\bigl( \lambda_{1}, \ldots, \lambda_{r} \bigr) \, h^{(j)}(t) = 0,
    \end{equation}
    is a homogeneous linear differential equation, with characteristic equation $\sum_{j=0}^{r} e_{r-j}\bigl( \lambda_{1}, \ldots, \lambda_{r} \bigr) \, \lambda^{j} = 0$. Given the property of the elementary symmetric polynomial, this characteristic equation can be equivalently expressed as $\prod_{i}^{r} (\lambda + \lambda_i) = 0$, with roots $\lambda_1, \ldots, \lambda_r$.\footnote{
    Take $r=2$ as an example: $\sum_{j=0}^{2} e_{2-j}\bigl( \lambda_{1}, \lambda_{2} \bigr) \, \lambda^{j} = \lambda^2 + (\lambda_1+\lambda_2)\lambda+\lambda_1\lambda_2=(\lambda+\lambda_1)(\lambda+\lambda_2)$, with roots $\lambda = -\lambda_1$ and $\lambda = -\lambda_2$.} Importantly, observe that \labelcref{eq:lower-bound-hocbf-general-form} is the unique solution for \labelcref{eq:equiv-hocbf-condition-equality}, i.e., $L_D \; h_{\mathrm{lb}}(t) = 0$ holds, where $L_D$ denotes a linear polynomial differential operator to an arbitrary continuously differentiable function $y(t)$ as $L_D\,y(t) \coloneqq \sum_{j=0}^{r} e_{r-j}\bigl( \lambda_{1}, \ldots, \lambda_{r} \bigr)\, y^{(j)}(t)$. Then, \labelcref{eq:equiv-hocbf-condition} is equivalent to $L_D \; h(t) \ge 0$. Introduce an auxiliary function $h_{\mathrm{aux}}(t) \coloneqq h(t) - h_{\mathrm{lb}}(t)$. By linearity of \( L_D \), $L_D \, h_{\mathrm{aux}}(t) = L_D\,h(t) - L_D\,h_{\mathrm{lb}}(t) \ge 0, \forall t \ge t_0$. Moreover, \( h_{\mathrm{aux}}(t_0)=h(t_0)-h_{\mathrm{lb}}(t_0)=0 \) since $h$ and $h_{\mathrm{lb}}$ have the same initial condition. By the maximum principle, it follows that \( h_{\mathrm{aux}}(t) \ge 0, \forall t \ge t_0 \), i.e., $h(t) \ge h_{\mathrm{lb}}(t), \forall t \ge t_0$. We conclude that $h_{\mathrm{lb}}(t)$ in \eqref{eq:lower-bound-hocbf-general-form} is the lower bound of $h(t)$.
\end{proof}

\begin{remark} \label{rem:hocbf}
Following \cref{cor:ct-lower-bound-hocbf}, the lower bound imposed by exponential \acp{hocbf} depends not only on the parameters of the linear class $\mathcal{K}$ functions but also on the initial conditions of the candidate \ac{cbf} and its time derivatives. On the other hand, \acp{hocbf} require that the initial value of each auxiliary function in \labelcref{eq:chain-hocbf} (except for the last one, $\Psi_{r}$) be non-negative \cite{xiao2019control}. To realize this, one can recursively show that sufficient (but not necessary) conditions are $\lambda_{i} \geq -h^{(i)} / {h^{(i-1)}}, \forall i \in \{1, \ldots, r-1\}$, and $\lambda_r > 0$. Therefore, the selection of the parameters is restricted by these initial conditions. These couplings can complicate control design, which motivates our approach presented next.
\end{remark}

\subsection{Truncated Taylor \acp{cbf}} \label{sec:taylor-based-alternative}
This section presents our \ac{ttcbf} approach to handle constraints with high relative degrees while requiring only one class $\mathcal{K}$ function.

Consider a \ac{cbf} with relative degree $r$. The fundamental discrete-time \ac{cbf} condition that we need to guarantee is \labelcref{eq:cbf-cond-dt}, i.e., $\Delta h \bigl(\bm{x}_k,\bm{u}_k\bigr) + \alpha \bigl( h(\bm{x}_{k}) \bigr) \ge 0, \forall\bm{x}_{k}\in C$. We approximate $\Delta h \bigl(\bm{x}_k,\bm{u}_k\bigr)$ with a truncated Taylor series with up to the $r$th time derivative of $h$, i.e., $\Delta h(\bm{x}_k, \bm{u}_k) \approx \Delta t \dot{h}(\bm{x}_k) + \frac{1}{2}\Delta t^2 \ddot{h}(\bm{x}_k) + \cdots + \frac{1}{r!}\Delta t^r h^{(r)}(\bm{x}_k,\bm{u}_k)$. Here, the $r$th derivative $h^{(r)}(\bm{x}_k,\bm{u}_k)$ captures the control input.

\begin{definition}[\aclp{ttcbf}] \label{def:taylor-cbf}
    Given a set $C$ as in \eqref{eq:safe-set}. An $(r+1)$th continuously differentiable function $h: \mathcal{X} \to \mathbb{R}$ is a candidate \acl{ttcbf} (\acs{ttcbf}) with relative degree $r$ for system \labelcref{eq:affine-sys-discrete} if there exists a class $\mathcal{K}$ function $\alpha$ satisfying $\alpha(z) \le z$ such that $\forall \bm{x}_k \in C$,
    \begin{equation} \label{eq:taylor-hocbf-cond}
    \begin{aligned}
        \sup_{\bm{u}_k \in \mathcal{U}} \bigl[ \Delta t \; \dot{h}(\bm{x}_k) + \cdots &+ \frac{1}{r!}\Delta t^r \; h^{(r)}(\bm{x}_k,\bm{u}_k) + \\
        & \alpha\bigl(h(\bm{x}_k)\bigr) \bigr] \ge \gamma \Delta t^{r+1},
    \end{aligned}
    \end{equation}
    Here, $\gamma$ is a design parameter that satisfies $\gamma \ge \frac{\Gamma_{r+1}}{(r+1)!}$, and we assume that the $(r+1)$th time derivative of $h$ is uniformly bounded over $C$, i.e., there exists a constant $\Gamma_{r+1} > 0$ such that $\left| h^{(r+1)}(\bm{x}) \right| \le \Gamma_{r+1}, \forall \bm{x} \in C$.
\end{definition}

\begin{theorem} \label{thm:taylor}
Given a \ac{ttcbf} in \cref{def:taylor-cbf} with associated set $C$ from \eqref{eq:safe-set}, any Lipschitz continuous controller satisfying \eqref{eq:taylor-hocbf-cond} renders $C$ forward invariant for system \eqref{eq:affine-sys-discrete}.
\end{theorem}

\begin{proof}
Because $h$ is $(r+1)$th continuously differentiable, Taylor's theorem guarantees the existence of a point $\bm{\xi}$ between $\bm{x}_k$ and $\bm{x}_{k+1}$ such that $h(\bm{x}_{k+1}) = h(\bm{x}_k) + \Delta t\, \dot{h}(\bm{x}_k) + \cdots + \frac{1}{r!}\Delta t^r h^{(r)}(\bm{x}_k,\bm{u}_k) + R_{r+1}$, with the remainder $R_{r+1} \coloneqq \frac{1}{(r+1)!}\Delta t^{r+1}\, h^{(r+1)}(\bm{\xi})$.
Since the $h^{(r+1)}$ is bounded by $\Gamma_{r+1}$, it follows that $\left| R_{r+1} \right| \le \frac{\Gamma_{r+1}}{(r+1)!}\Delta t^{r+1}$. Thus, we have $h(\bm{x}_{k+1}) \ge h(\bm{x}_k) + \Delta t\, \dot{h}(\bm{x}_k) + \cdots + \frac{1}{r!}\Delta t^r\, h^{(r)}(\bm{x}_k) - \frac{\Gamma_{r+1}}{(r+1)!}\Delta t^{r+1}$. Given \labelcref{eq:taylor-hocbf-cond}, one easily obtains
\begin{equation} \label{eq:taylor-cond-proof}
    h(\bm{x}_{k+1}) \ge h(\bm{x}_k) - \alpha\bigl(h(\bm{x}_k)\bigr) + \left(\gamma - \frac{\Gamma_{r+1}}{(r+1)!}\right) \Delta t^{r+1}.
\end{equation}
Provided that $\gamma \ge \frac{\Gamma_{r+1}}{(r+1)!}$, we have $h(\bm{x}_{k+1}) \ge h(\bm{x}_k) - \alpha\bigl(h(\bm{x}_k)\bigr)$, i.e., $\Delta h(\bm{x}_k) + \alpha\bigl(h(\bm{x}_k)\bigr) \ge 0$, rendering \labelcref{eq:taylor-hocbf-cond} to be a sufficient condition of the discrete-time \ac{cbf} condition \labelcref{eq:cbf-cond-dt}. The remaining proof follows \cite[Thm. 1]{ahmadi2019safe}.
\end{proof}

The design parameter $\gamma$ in \labelcref{eq:taylor-hocbf-cond} considers the Taylor series truncation error. We will discuss how to handle it in \cref{sec:discussions}.

\section{Numerical Experiments} \label{sec:experiments}

We revisit the collision-avoidance example in \cref{fig_example_class_k}. In \cref{sec:qp-formulation}, we formulate the \ac{qp} problem using both the standard \ac{hocbf} approach and our \ac{ttcbf} approach. We present the experimental results of the standard \ac{hocbf} approach in \cref{sec:first-experiment} and those of our approach in \cref{sec:second-experiment}. Code reproducing the experimental results is available at our open-source  repository\footnote{\href{https://github.com/bassamlab/sigmarl}{github.com/bassamlab/sigmarl}} \cite{xu2024sigmarl}.

% Table for parameters
\begin{table}[t]
    \caption{Parameters used in the experiments.}
    \centering
    \begin{threeparttable}
        \begin{tabular}{ll}
            \toprule
            Parameters & Values \\
            \midrule
            Robot and obstacle radii $r, r_{\mathrm{obs}}$ & \SI{1}{\meter}, \SI{2}{\meter} \\
            Obstacle position $(x_{\mathrm{obs}}, y_{\mathrm{obs}})$ & $(\SI{0}{\meter}, \SI{-3.1}{\meter})$ \\
            Robot initial position & $(\SI{-10}{\meter}, \SI{0}{\meter})$ \\
            Robot initial speed, acceleration & \SI{10}{\meter\per\second}, \SI{0}{\meter\per\second\squared} \\
            Reference $y_{\mathrm{ref}}, v_{x,\mathrm{ref}}, v_{y,\mathrm{ref}}$ & \SI{0}{\meter}, \SI{10}{\meter\per\second}, \SI{0}{\meter\per\second}  \\
            Penalty parameters $p_{v_x}, p_{v_y}, p_{y}$ & 1, 1, 1000 \\
            Admissible control input $\bm{u}_{\min}, \bm{u}_{\max}$ & $\pm \SI{1000}{\meter\per\second\squared}$ \\
            Sampling period $\Delta t$, simulation duration & \SI{0.01}{\second}, \SI{2}{\second} \\
            Term $\gamma \Delta t^{r+1}$ in \eqref{eq:taylor-hocbf-cond} & $\approx 0$ \\
            \bottomrule
        \end{tabular}
    \end{threeparttable}
    \label{tab:parameters}
\end{table}

\subsection{\acl{qp} Formulation} \label{sec:qp-formulation}
To facilitate the computation of time derivatives, we select the square of the distance between the robot and the obstacle as the candidate \ac{cbf}, i.e.,
\begin{equation} \label{eq:eva-cbf}
    h\bigl(\bm{x}(t)\bigr) = \bigl( x(t) - x_{\mathrm{obs}} \bigr)^2 + \bigl(y(t) - y_{\mathrm{obs}}\bigr)^2 - \bigl(r + r_{\mathrm{obs}}\bigr)^2,
\end{equation}
where $\bm{x}(t) \coloneqq \bigl(x(t), y(t)\bigr)^\top \in \mathbb{R}^2$ and $(x_{\mathrm{obs}}, y_{\mathrm{obs}}) \in \mathbb{R}^2$ denote the positions of the robot and the obstacle, respectively. The radii of the robot and the obstacle are given by $r > 0$ and $r_{\mathrm{obs}} > 0$, respectively. Clearly, if $h\bigl(\bm{x}(t)\bigr) \ge 0$ for all $t > 0$, the robot remains safe. We model the robot dynamics as a double integrator and choose the accelerations along the $x$- and $y$-directions as the control inputs, denoted by $\bm{u}(t) \coloneqq \bigl(u_{x}(t), u_{y}(t)\bigr)^\top$. Therefore, the candidate \ac{cbf} in \eqref{eq:eva-cbf} has relative degree two. Since the obstacle position is constant, the first- and second-time derivatives of $h$ (omitting the explicit time dependence for brevity) are
\begin{align*}
    \dot{h} &= 2\Bigl(x - x_{\mathrm{obs}}\Bigr)\dot{x} + 2\Bigl(y - y_{\mathrm{obs}}\Bigr)\dot{y}, \\
    \ddot{h} &= 2\Bigl(\dot{x}^2 + \dot{y}^2 + \Bigl(x - x_{\mathrm{obs}}\Bigr) u_{x} + \Bigl(y - y_{\mathrm{obs}}\Bigr) u_{y}\Bigr).
\end{align*}

We conduct the experiments in simulation using discrete time. We use a sampling period of $\Delta t = \SI{0.01}{\second}$ and set the duration of each simulation to \SI{2}{\second}. Let $k \in \mathbb{N}$ denote the current time step. We formulate the \ac{qp} problem as
\begin{subequations} \label{eq:qp}
    \begin{align}
        \bm{u}_k  &= \arg\min_{\bm{u}_k} \; p_{v_x} \Bigl(\hat{v}_{x,k+1} - v_{x,\mathrm{ref}}\Bigr)^2 + \notag \\
        &\quad \quad p_{v_y} \Bigl(\hat{v}_{y,k+1} - v_{y,\mathrm{ref}}\Bigr)^2 + p_{y} \Bigl(\hat{y}_{k+1} - y_{\mathrm{ref}}\Bigr)^2, \label{eq:qp-cost-fcn} \\
        \text{s.t.} &\quad \ddot{h}_k + (\lambda_1 + \lambda_2) \dot{h}_k + \lambda_1 \lambda_2 h_k \ge 0 \ (\text{see } \eqref{eq:ct-hocbf-cond-flatten})  \quad \text{or} \notag \\
        & \quad \Delta t \dot{h}_k + \frac{1}{2} \Delta t^2 \ddot{h}_k + \lambda_1 h_k \ge \gamma \Delta t^3 \ (\text{see } \eqref{eq:taylor-hocbf-cond}), \label{eq:qp-safe-constraint} \\
        &\quad \bm{u}_{\min} \leq \bm{u}_k \leq \bm{u}_{\max}. \label{eq:qp-u-limit}
    \end{align}
\end{subequations}
The cost function in \eqref{eq:qp-cost-fcn} comprises three terms. The first two terms penalize deviations from the reference speeds, $v_{x,\mathrm{ref}}$ along the $x$-axis and $v_{y,\mathrm{ref}}$ along the $y$-axis. The predicted next-step speed along the $x$-axis is computed as $\hat{v}_{x,k+1}(u_{x,k}) = v_{x,k} + \Delta t \, u_{x,k}$, under a zero-order hold assumption of the control input, where $v_{x,k}$ denotes the current $x$-speed. Similarly, $\hat{v}_{y,k+1}(u_{y,k}) = v_{y,k} + \Delta t \, u_{y,k}$, and $v_{y,\mathrm{ref}}$ is the reference speed along the $y$-axis. In addition, to encourage the robot to follow its reference path, which is a horizontal line defined by $y = y_{\mathrm{ref}}$, we penalize the deviation from it. The predicted next-step $y$-position is computed as $\hat{y}_{k+1}(u_{y,k}) = y_{k} + \Delta t \, v_{y,k} + \tfrac{1}{2}\Delta t^2 \, u_{y,k}$. The penalty parameters are denoted by $p_{v_x}$, $p_{v_y}$, and $p_{y}$. Constraint \eqref{eq:qp-safe-constraint} is the \ac{cbf} condition that ensures safety and is computed as \eqref{eq:ct-hocbf-cond-flatten} when using the standard \ac{hocbf} approach or as \eqref{eq:taylor-hocbf-cond} when using our approach. We use linear class $\mathcal{K}$ functions for both approaches. Since the \ac{cbf} \eqref{eq:eva-cbf} has relative degree two, the \ac{hocbf} approach requires two class $\mathcal{K}$ functions, with their parameters denoted by $\lambda_1$ and $\lambda_2$. Our approach requires only one, with its parameter denoted by $\lambda_1$. For simplicity, we ignore the term $\gamma \Delta t^3$ in \eqref{eq:qp-safe-constraint} in the case of our approach. Constraint \eqref{eq:qp-u-limit} imposes admissible control input, which is set to sufficiently high since they are not the focus of our experiments. We solve \eqref{eq:qp} iteratively at each time step $k$ and apply the solution $\bm{u}_k$.

\begin{figure*}[t]
    \centering
    \begin{subfigure}{0.34\linewidth}
        \centering
        \includegraphics[width=\linewidth]{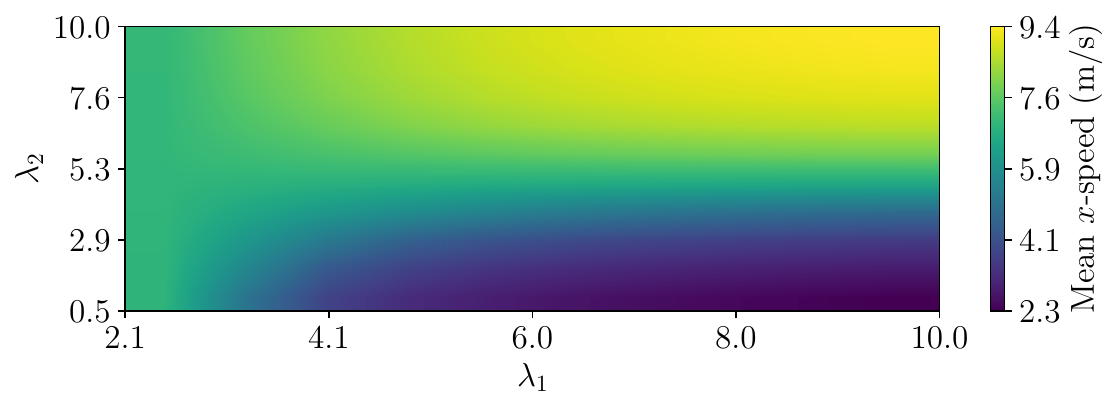}
        \caption{Mean $x$-speed.}\label{fig_heatmap_hocbf}
    \end{subfigure}
    \hfill
    \begin{subfigure}{0.30\linewidth}
        \centering
        \includegraphics[width=\linewidth]{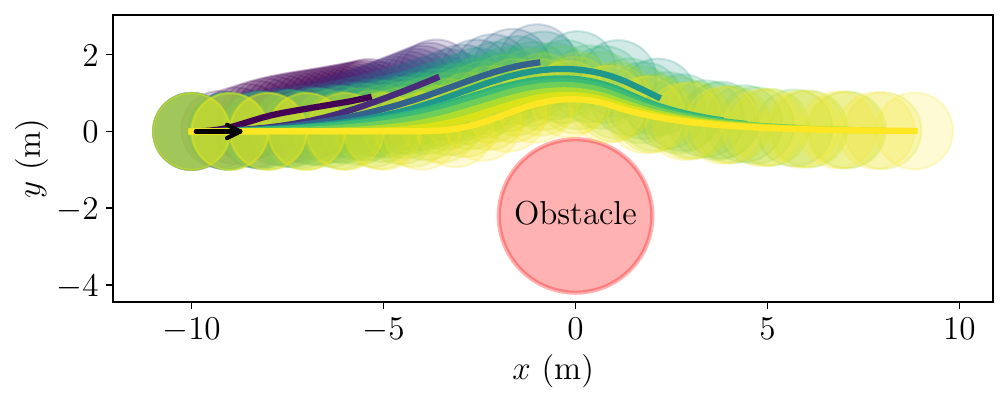}
        \caption{Footprint.} \label{fig_footprint_hocbf}
    \end{subfigure}
    \hfill
    \begin{subfigure}{0.33\linewidth}
        \centering
        \includegraphics[width=\linewidth]{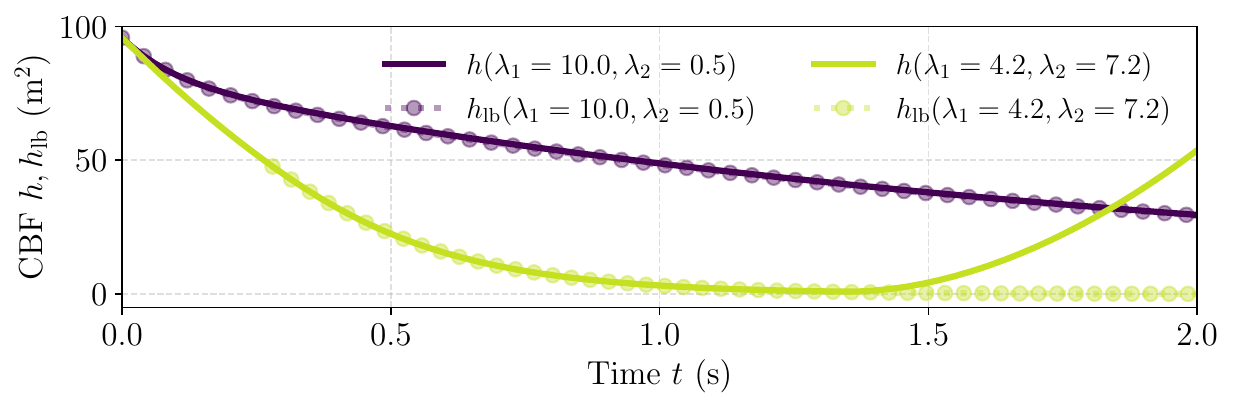}
        \caption{\ac{cbf} value $h$ and its analytical lower bound $h_{\mathrm{lb}}$ computed by \eqref{eq:lower-bound-hocbf-general-form}.}\label{fig_h_hocbf}
    \end{subfigure}
    \caption{Experimental results from the standard \ac{hocbf} approach with different parameters for the class $\mathcal{K}$ functions.}
    \label{fig_hocbf}
\end{figure*}

\begin{figure*}[t]
    \centering
    \begin{subfigure}{0.34\linewidth}
        \centering
        \includegraphics[width=\linewidth]{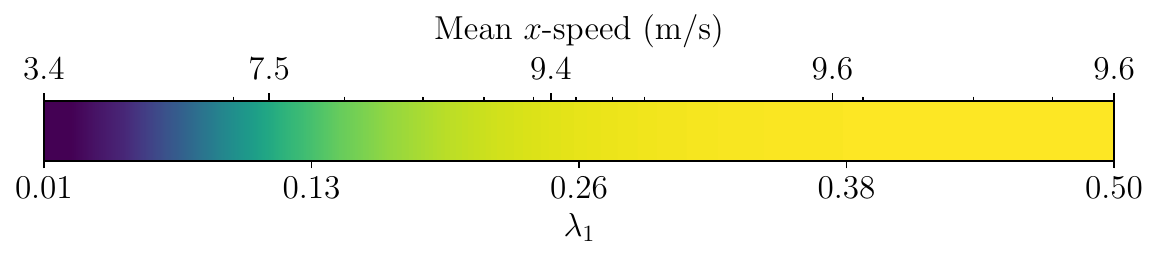}
        \caption{Mean $x$-speed.}\label{fig_heatmap_taylor}
    \end{subfigure}
    \hfill
    \begin{subfigure}{0.30\linewidth}
        \centering
        \includegraphics[width=\linewidth]{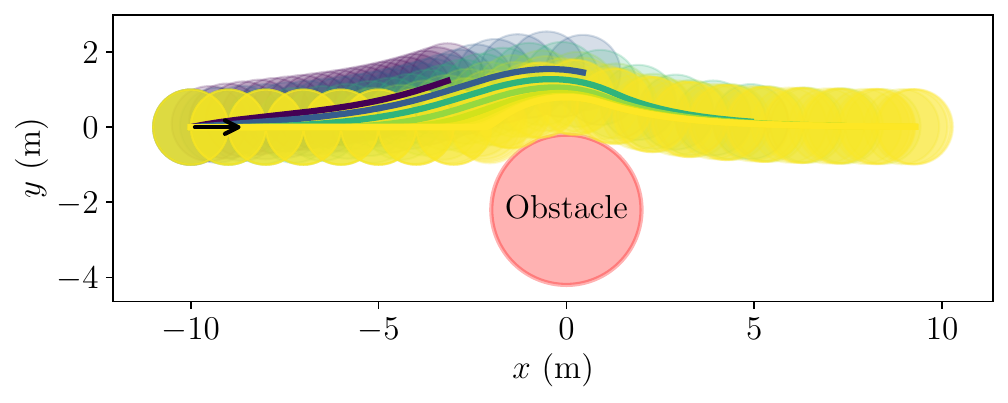}
        \caption{Footprint.} \label{fig_footprint_taylor}
    \end{subfigure}
    \hfill
    \begin{subfigure}{0.33\linewidth}
        \centering
        \includegraphics[width=\linewidth]{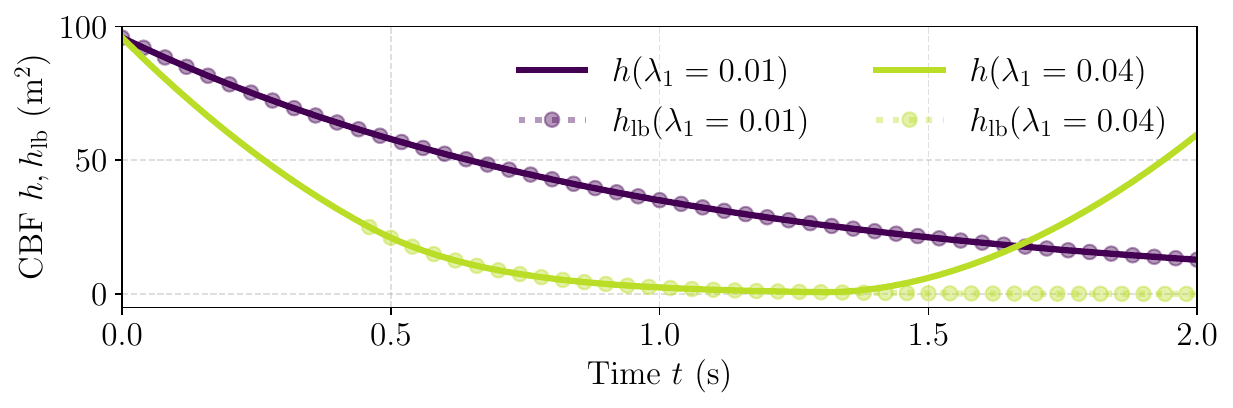}
        \caption{\ac{cbf} value $h$ and its analytical lower bound $h_{\mathrm{lb}}$ computed by \eqref{eq:lower-bound-de-cbf}.}\label{fig_h_taylor}
    \end{subfigure}
    \caption{Experimental results from our Taylor-based approach with different parameters for the class $\mathcal{K}$ function.}
    \label{fig_taylor}
\end{figure*}

\subsection{First Experiment: \ac{hocbf} Approach} \label{sec:first-experiment}
In this experiment, we apply the standard \ac{hocbf} approach and evaluate its performance under different parameter settings of the class $\mathcal{K}$ functions. We run simulations over a dense grid of parameter pairs, with $(\lambda_1, \lambda_2) \in [2.1, 10] \times [0.5,\, 10]$. Note that we choose $\lambda_{1, \min} = 2.1$ because the \ac{hocbf} approach requires $\lambda_1 \ge -\dot{h}_0 / h_0$, as discussed in \cref{rem:hocbf}. From the given initial conditions (see \cref{tab:parameters}), we compute $-\dot{h}_0 / h_0 \approx 2.10$. We empirically set the other three bounds, $\lambda_{1, \max}$, $\lambda_{2, \min}$, and $\lambda_{2, \max}$, to be sufficiently low or high to cover a broad range of control behaviors. With our \ac{qp} formulation in \eqref{eq:qp}, the controller decelerates the robot when the \ac{cbf} condition is active, especially in the $x$-direction. Therefore, we use the mean $x$-speed during each simulation as an indicator of the controller's performance. Using the total mean speed yields similar results.

One typical statement in \acp{cbf} is that a linear class $\mathcal{K}$ function with a larger parameter is less conservative, and vice versa \cite{zeng2021safetycriticala}. However, our experimental results show that this does not necessarily hold in \acp{hocbf}. \Cref{fig_heatmap_hocbf} shows the mean $x$-speed for all parameter pairs. A higher mean speed indicates a less conservative controller. As observed, if $\lambda_1 > 3.0$, the achieved mean $x$-speed increases with $\lambda_2$; otherwise, it remains unchanged regardless of $\lambda_2$. In contrast, if $\lambda_2 < 5.3$, increasing $\lambda_1$ reduces the achieved mean $x$-speed, and vice versa. These observations can be explained by our analytical lower bound \eqref{eq:lower-bound-hocbf-general-form}. Although the decay constant $\lambda_i$ of each sum term matches the parameter of one of the linear class $\mathcal{K}$ functions, its coefficient $c_i$ is jointly determined by all these parameters and the initial conditions (see \eqref{eq:coefficients}), preventing the direct application of the above statement in \acp{hocbf}.

\Cref{fig_footprint_hocbf} shows the robot footprints for each parameter pair, using the same color coding as in \cref{fig_heatmap_hocbf}. Some parameter pairs cause the robot to fail to bypass the obstacle within the simulation duration, because the imposed lower bound on the \ac{cbf} value is high, forcing a low decay rate of $h(t)$. \Cref{fig_h_hocbf} illustrates this with two representative parameter pairs. The solid lines represent actual \ac{cbf} values, and the dotted lines with circle markers represent analytical lower bounds computed by \eqref{eq:lower-bound-hocbf-general-form}. Consider the most conservative parameter pair with $\lambda_1 = 10$ and $\lambda_2 = 0.5$. Applying \eqref{eq:lower-bound-hocbf-general-form} with the coefficients $\bm{c}$ obtained by solving \eqref{eq:coefficients} under the initial condition $t_0 = \SI{0}{\second}$ (i.e., $h_{\mathrm{init}} = [h_{t_0}, \dot{h}_{t_0}]^\top = [95.8, -200]^\top$) yields
\[
    \bm{c} = M_{\lambda}^{-1} \, \bm{h}_{\mathrm{init}}
    = 
    \begin{bmatrix}
        1 & 1\\
        -10 & -0.5
    \end{bmatrix}^{-1}
    \begin{bmatrix}
        95.8\\
        -200
    \end{bmatrix}
    =
    \begin{bmatrix}
        16.0\\
        79.8
    \end{bmatrix},
\]
and thus $h_{\mathrm{lb}}(t) = 16.0e^{-10t} + 79.8e^{-0.5t}, \forall t \ge 0$.
This lower bound $h_{\mathrm{lb}}(t)$ prevents $h(t)$ from decaying below it. Since $h(t)$ would drop under $h_{\mathrm{lb}}(t)$ if the robot moves too quickly, the controller reduces its speed to avoid this violation. \Cref{fig_h_hocbf} visualizes this lower bound in dark purple, and we see that it closely matches the actual lower bound, which validates \cref{cor:ct-lower-bound-hocbf}. Note that for the other parameter pair, we compute the lower bound with the initial conditions at $t_0 \approx \SI{0.3}{\second}$, which is the moment when the \ac{cbf} condition becomes active and $h(t)$ reaches the actual lower bound. The lower bound computed at an earlier time would only loosely bound $h(t)$ and could not convincingly validate \cref{cor:ct-lower-bound-hocbf}.

\subsection{Second Experiment: Our \ac{ttcbf} Approach} \label{sec:second-experiment}
In this experiment, we apply our \ac{ttcbf} approach proposed in \cref{sec:taylor-based-alternative}. Similar to the standard discrete-time \ac{cbf} approach, we require $\lambda_1 \in (0, 1]$. We run simulations over a dense grid of $\lambda_1 \in [0.01, 0.50]$, which covers a broad range of control behaviors in this experiment.

\Cref{fig_heatmap_taylor} shows the mean $x$-speed for each $\lambda_1$. Straightforwardly, larger values of $\lambda_1$ lead to higher mean $x$-speed, simplifying the design compared to the standard \ac{hocbf} approach that requires the joint tuning of multiple class $\mathcal{K}$ functions. \Cref{fig_footprint_taylor} shows the robot footprints, and \Cref{fig_h_taylor} selectively shows the \ac{cbf} values and their corresponding lower bounds for two distinct parameter settings. The lower bounds, shown as dotted curves with circle markers, are computed using \eqref{eq:lower-bound-de-cbf}, which is derived from standard discrete-time \acp{cbf}. As observed, our approach imposes a lower bound that matches this analytical lower bound.

\section{Discussions} \label{sec:discussions}
The experiment in \cref{sec:first-experiment} demonstrates that the performance of the \ac{hocbf} approach is jointly influenced by the parameters of the class $\mathcal{K}$ functions, and tuning even two parameters can be challenging. This complexity can grow significantly for constraints with higher relative degrees, which require more class $\mathcal{K}$ functions. Our \ac{ttcbf} approach addresses this challenge by requiring only one class $\mathcal{K}$ function. While the standard \ac{cbf} approach remains preferable for constraints with relative degree one, the advantages of our approach become evident for higher relative degrees, where it simplifies the tuning process compared to the standard \ac{hocbf} approach.

Our \ac{ttcbf} approach uses a truncated Taylor series to approximate the forward difference $\Delta h$ in the discrete-time \ac{cbf} condition \eqref{eq:cbf-cond-dt}, leading to a more constrained condition due to the additional term $\left(\gamma - \frac{\Gamma_{r+1}}{(r+1)!}\right) \Delta t^{r+1}$ in \eqref{eq:taylor-cond-proof}. This inherently introduces conservatism. However, note that this term is scaled by $\Delta t^{r+1}$. 
When $\Delta t^{r+1}$ is sufficiently small, either due to small $\Delta t$ or large $r$, this term becomes negligible. In our experiment in \cref{sec:second-experiment}, with $\Delta t = \SI{0.01}{\second}$ and $r=2$, omitting this term did not affect control performance. A sampling period of \SI{0.01}{\second} was feasible because the \ac{qp} problem \eqref{eq:qp} could be solved efficiently, averaging \SI{3}{\milli\second} per step. If this term is not negligible (for example in computationally intensive applications that require larger sampling periods), one needs to determine the design parameter $\gamma$. By \cref{def:taylor-cbf}, this parameter should be greater than or equal to the upper bound of the $(r+1)$th time derivative of the candidate \ac{cbf}, denoted by $\Gamma_{r+1}$. If this upper bound cannot be directly estimated, one can empirically select a constant $\gamma$ that yields satisfying control performance.

% This may introduce conservatism since our approach accounts for 

% Numerous studies can benefit from our approach, including but not restricted to penalty methods \cite{xiao2019control}, adaptive \acp{cbf} \cite{xiao2022adaptive}, and learning-based class $\mathcal{K}$ functions \cite{ma2022learning, kim2025learning}.

%===============================================================================
\section{Conclusions} \label{sec:conclusions}
In this work, we have derived the explicit form of the condition imposed by the standard \ac{hocbf} approach, expressed as a homogeneous linear differential inequality. We have demonstrated that the solution of the corresponding equality represents the imposed lower bound on the \ac{cbf} value, showing explicitly how the parameters of class $\mathcal{K}$ functions and the initial conditions of the candidate \ac{cbf} determine this lower bound. We proposed our \ac{ttcbf} approach. It is enabled by using a truncated Taylor series that approximates the discrete-time \ac{cbf} condition. It simplifies control design for constraints with high relative degrees by requiring only one class $\mathcal{K}$ function compared to the \ac{hocbf} approach, which requires multiple class $\mathcal{K}$ functions. Numerical experiments in a collision-avoidance scenario confirmed that the derived analytical lower bound closely matches the actual ones. They also validated our approach's ability to handle constraints with high relative degrees while using only one class $\mathcal{K}$ function.

\bibliographystyle{IEEEtran}
\bibliography{00_literature}

\end{document}